%% file: HDFOL.tex
\documentclass[a4paper,UKenglish,cleveref,autoref]{lipics-v2019}

\input{preamble}

\title{Horn Clauses in Hybrid-Dynamic First-Order Logic}

\author{Daniel G\u{a}in\u{a}}{%
  Institute of Mathematics for Industry, Kyushu University, Japan \and
  Department of Mathematics and Statistics, La Trobe University, Australia%
}{%
  daniel@imi.kyushu-u.ac.jp%
}{}{}

\author{Ionu\c{t} \c{T}u\c{t}u}{
  Simion Stoilow Institute of Mathematics of the Romanian Academy, Romania \and
  Department of Computer Science, Royal Holloway University of London, UK%
}{%
  ittutu@gmail.com%
}{}{}

\authorrunning{D.\ G\u{a}in\u{a} and I.\ \c{T}u\c{t}u}

\Copyright{Daniel G\u{a}in\u{a} and Ionu\c{t} \c{T}u\c{t}u}

\ccsdesc[500]{Theory of computation~Modal and temporal logics}
\ccsdesc[300]{Theory of computation~Logic and verification}
\keywords{Hybrid logic, Dynamic logic, Horn clause, Herbrand's theorem}

\nolinenumbers % uncomment to disable line numbering

\hideLIPIcs % uncomment to remove references to LIPIcs series (logo, DOI, ...)

\begin{document}

\maketitle

\begin{abstract}
  We propose a hybrid-dynamic first-order logic as a formal foundation for specifying and reasoning about reconfigurable systems.
  As the name suggests, the formalism we develop extends (many-sorted) first-order logic with features that are common to hybrid logics and to dynamic logics.
  This provides certain key advantages for dealing with reconfigurable systems, such as:
  (a)~a signature of nominals, including operation and relation symbols, that allows references to specific possible worlds / system configurations -- as in the case of hybrid logics;
  (b)~distinguished signatures of rigid and flexible symbols, where the rigid symbols are interpreted uniformly across possible worlds; this supports a rigid form of quantification, which ensures that variables have the same interpretation regardless of the possible world where they are evaluated;
  (c)~hybrid terms, which increase the expressive power of the logic in the context of rigid symbols; and
  (d)~modal operators over dynamic-logic actions, which are defined as regular expressions over binary nominal relations.
  We then study Horn clauses in this hybrid-dynamic logic, and develop a series of results that lead to an initial-semantics theorem for arbitrary sets of clauses.
  This shows that a significant fragment of hybrid-dynamic first-order logic has good computational properties, and can serve as a basis for defining executable languages for reconfigurable systems.
  Lastly, we set out the foundations of logic programming in this fragment by proving a hybrid-dynamic variant of Herbrand's theorem, which reduces the semantic entailment of a logic-programming query by a program to the search of a suitable answer substitution.
\end{abstract}

%%%%%%%%%%%%%%%%%%%%%%%%%%%%%%%%%%%%%%%%%%%%%%%%%%%%%%%%%%%%%%%%%%%%%%%%%%%%%%%% 

\section{Introduction}

The dynamic-reconfiguration paradigm is a most promising approach in the development of highly complex and integrated systems of interacting `components', which now often evolve dynamically, at run time, in response to internal or external stimuli.
More than ever, we are witnessing a continuous increase in the number of applications with reconfigurable features, many of which have aspects that are safety- or security-critical.
This calls for suitable formal-specification and verification technologies, and there is already a significant body of research on this topic; hybrid(ized) logics~\cite{Brauner11,MartinsMDB11}, first-order dynamic logic~\cite{HarelKT01}, and modal \(\mu\)-calculus~\cite{GrooteM2014} are three prominent examples, among many others.

In this paper, we focus on those reconfigurable systems
whose states or configurations can be presented explicitly, based on some kind of context-independent data types,
and for which we distinguish the computations performed at the local/configuration level from the dynamic evolution of the configurations.
This suggests a two-layered approach to the design and analysis of reconfigurable systems, involving
\emph{a local view}, which amounts to describing the structural properties of configurations, and
\emph{a global view}, which corresponds to a language for specifying and reasoning about the way system configurations evolve.

For that purpose, we propose a new modal logical system for the reconfiguration paradigm that is obtained by enriching first-order logic (regarded as a base logic, or parameter for the whole construction) with both hybrid and dynamic features.
More specifically, we model reconfigurable systems as Kripke structures (or transition systems), where:
\begin{itemize}
  
\item from a local perspective, we consider a dedicated first-order signature for configurations, and hence capture configurations as first-order structures for that signature; and

\item from a global perspective, we consider a second first-order signature for the possible worlds of the Kripke structure;
  the terms over that signature are \emph{nominals} used to identify configurations, and the binary nominal relations are regarded as \emph{modalities}, which capture the transitions, or accessibility relations, between configurations.
  
\end{itemize}
Concerning the syntax of the logic proposed, sentences are build from equations and relational atoms over the two first-order signatures mentioned above (one pertaining to data, and the other to possible worlds) by using Boolean connectives, quantifiers, standard hybrid-logic operators such as \emph{retrieve} and \emph{store}, and dynamic-logic operators such as \emph{necessity} over structured actions, which are defined as regular expressions over modalities.

The construction is reminiscent of the hybridization of institutions from~\cite{MartinsMDB11,DiaconescuM16} and of the hybrid-dynamic logics presented in~\cite{BohrerP18,HennickerMK19}, but it departs fundamentally from any of those studies due to the fact that the possible worlds of the Kripke structures that we consider here have an algebraic structure.
This special feature of the logic that we put forward is extremely important for dealing with reconfigurable systems whose states are obtained from initial configurations by applying constructor operations; see, e.g.~\cite{GainaTR18}.
In this context, we advance a general notion of Horn clause, which allows the use of implications, universal quantifiers, as well as the hybrid- and dynamic-logic operators listed above.
We study congruences over Kripke structures, and show how these can be used in quotienting initial Kripke models with respect to Horn clauses.
Initial semantics has received a lot of attention in the formal-specification literature, where it has traditionally been linked to formalisms that support execution by means of rewriting.
Here, we explore the key role of initiality in logic programming -- where the initial models are typically referred to as least Herbrand models -- and we prove a hybrid-dynamic version of Herbrand's theorem.
This a fundamental result, because it establishes a connection between the semantic entailment of a logic-programming query and the existence of an answer substitution for that query -- which is a syntactic construct, and can subsequently be obtained by proof-theoretic means.

The paper is structured as follows:
Section~\ref{section:HDFOLS} is devoted to the definition of the new logic, which we call hybrid-dynamic first-order logic with user-defined sharing (this last attribute is meant to indicate the fact that users have control over the symbols that should be interpreted the same across the possible worlds of a Kripke structure);
then, in Section~\ref{section:initiality} we introduce Horn clauses and prove the main initiality results of the paper (for atomic sentences and for Horn clauses), while in Section~\ref{section:Herbrand-theorem} we focus on queries and Herbrand's theorem;
lastly, in Sections~\ref{section:related-work} and~\ref{section:conclusions} we discuss related work and potential future research.

%%%%%%%%%%%%%%%%%%%%%%%%%%%%%%%%%%%%%%%%%%%%%%%%%%%%%%%%%%%%%%%%%%%%%%%%%%%%%%%% 

\section{Hybrid-Dynamic First-Order Logic with user-defined Sharing}
\label{section:HDFOLS}

The hybrid-dynamic first-order logic (\(\HDFOLS\)) that we propose in this paper is rooted in the same ideas that underlie hybrid first-order logic~\cite{Brauner11} and hybrid first-order logic with rigid symbols~\cite{DiaconescuM16,Diaconescu16}. 
Its signatures contain distinguished subsets of so-called \emph{rigid symbols}, which may be sorts, operation or relation symbols, and are meant to have the same interpretation across the possible worlds of a given Kripke structure.
This provides support for the most common form of quantification in modal logic, where the semantics of variables is rigid.
We develop these ideas further by considering first-order structures of possible worlds.

We present \(\HDFOLS\) from an institutional perspective~\cite{GoguenB92}, meaning that we focus on signatures and signature morphisms (though, for the purpose of this paper, inclusions would suffice), Kripke structures and homomorphisms, sentences, and the (local) satisfaction relation and condition that relate the syntax and the semantics of \(\HDFOLS\).
However, other that the notations used, the text requires no prior knowledge of institution theory, and should be accessible to readers with a general background in modal and first-order model theory.

\minisec{Signatures}
The \emph{signatures} of \(\HDFOLS\) are tuples of the form \(\Delta = (\Sigma^{\nominal}, \Sigma^{\rigid} \subseteq \Sigma)\), where:
\begin{enumerate}
  
\item \(\Sigma^{\nominal} = (S^{\nominal}, F^{\nominal}, P^{\nominal})\) is a first-order signature of \emph{nominals} such that \(S^{\nominal} = \{\star\}\) is a singleton,

\item \(\Sigma^{\rigid} = (S^{\rigid}, F^{\rigid}, P^{\rigid})\) is a first-order signature of so-called \emph{rigid symbols}, and

\item \(\Sigma = (S, F, P)\) is a first-order signature of both \emph{rigid} and \emph{flexible} symbols.
  
\end{enumerate}
We denote by \(F^{\flexible}\) and \(P^{\flexible}\) the sub-families of \(F\) and \(P\), respectively, that consist of \emph{flexible symbols} (obtained by removing the rigid symbols).
Throughout the paper, we generally denote by \(\Delta\) and \(\Delta'\) signatures of the form \((\Sigma^{\nominal}, \Sigma^{\rigid} \subseteq \Sigma)\) and \((\Sigma'^{\nominal}, \Sigma'^{\rigid} \subseteq \Sigma')\), respectively.

\minisec{Signature morphisms}
A \emph{signature morphism} \(\varphi \colon \Delta \to \Delta'\) consists of a pair of first-order signature morphisms
\(\varphi^{\nominal} \colon \Sigma^{\nominal} \to \Sigma'^{\nominal}\) and \(\varphi \colon \Sigma \to \Sigma'\) such that \(\varphi(\Sigma^{\rigid}) \subseteq \Sigma'^{\rigid}\).

\minisec{Kripke structures}
The \emph{models} of a signature \(\Delta\) are pairs \((W, M)\), where:
\begin{enumerate}
  
\item \(W\) is a \(\Sigma^{\nominal}\)-model, for which we denote by \(|W|\) the carrier set of the only sort \(\star\) in \(\Sigma^{\nominal}\), 

\item \(M = \{M_{w}\}_{w \in |W|}\) is a family of \(\Sigma\)-models, indexed by \emph{possible worlds} \(w \in |W|\), such that the rigid symbols have the same interpretation across the worlds;
  that is, \(M_{w_{1}, \varsigma} = M_{w_{2}, \varsigma}\) for all possible worlds \(w_{1}, w_{2} \in |W|\) and all symbols \(\varsigma\) in \(\Sigma^{\rigid}\).
  
\end{enumerate}
\vspace{-\smallskipamount}

\minisec{Kripke homomorphisms}
A \emph{morphism} \(h \colon (W, M) \to (W', M')\) is also a pair, given by
a \(\Sigma^{\nominal}\)-homomorphism \(h \colon W \to W'\) and,
for each \(w \in |W|\), a \(\Sigma\)-homomorphism \(h_{w} \colon M_{w} \to M'_{h(w)}\),
such that \(h_{w_{1}, s} = h_{w_{2}, s}\) for all possible worlds \(w_{1}, w_{2} \in |W|\) and all rigid sorts \(s \in S^{\rigid}\).

\minisec{Actions}
As in dynamic logic, \(\HDFOLS\) supports structured actions obtained from atoms using sequential composition, union, and iteration.
The set \(A^{\nominal}\) of \emph{actions} over \(\Sigma^{\nominal}\) is defined inductively, according to the following grammar:
\(\act \Coloneqq \lambda \in P^{\nominal}_{\star \star} \mid \act \comp \act \mid \act + \act \mid \act^{*}\).

Actions are interpreted in Kripke structures as \emph{accessibility relations} between possible worlds.
This is done by extending the interpretation of binary modalities (symbols in \(P^{\nominal}_{\star \star}\)) as follows:
\(W_{\act_{1} \comp \act_{2}} = W_{\act_{1}} \comp W_{\act_{2}}\) (diagrammatic composition of relations), 
\(W_{\act_{1} + \act_{2}} = W_{\act_{1}} \cup W_{\act_{2}}\) (union of relations), and
\(W_{\act^{*}} = (W_\act)^{*}\) (reflexive and transitive closure of a relation).

\minisec{Hybrid terms}
For every \(\Sigma^{\nominal}\)-model \(W\),
the family \(T^{W} = \{ T^{W}_{w} \}_{w \in |W|}\) of hybrid terms over \(W\) is defined inductively as the least family of (sorted) sets satisfying the following rules:\footnote{For brevity, we often use \(\tau\) or \(t\) to denote not only terms, but also tuples of terms.}

\medskip
\noindent
\parbox{12.25em}{%
  \textsf{(1)}~\inferrule{w \in |W| \and \tau \in T^{W}_{w, \ari}}{\sigma(\tau) \in T^{W}_{w, s}}\\[1ex]
  \phantom{\textsf{(1)}} where \(\sigma \in F^{\rigid}_{\ari \to s}\)
}
\hfill
\parbox{12.25em}{%
  \textsf{(2)}~\inferrule{w \in |W| \and \tau \in T^{W}_{w, \ari}}{\sigma(w, \tau) \in T^{W}_{w, s}}\\[1ex]
  \phantom{\textsf{(2)}} where \(\sigma \in F^{\flexible}_{\ari \to s}\)
}
\hfill
\parbox{12.25em}{%
  \textsf{(3)}~\inferrule{w_{0}, w \in |W| \and \tau \in T^{W}_{w_{0}, s}}{\tau \in T^{W}_{w, s}}\\[1ex]
  \phantom{\textsf{(3)}} where \(s \in S^{\rigid}\)
}
\medskip

\noindent Notice that flexible operation symbols receive a possible world \(w \in |W|\) as an extra argument while the rigid operation symbols keep their initial arity.
Also, hybrid terms of rigid sorts are shared across the worlds.
These are important considerations for Proposition~\ref{proposition:tm-reachability} below.

Given a world \(w \in |W|\), the \(S\)-sorted set \(T^{\Delta}_{w}\) can be regarded as a \(\Sigma\)-model by interpreting
every rigid operation symbol \(\sigma \colon \ari \to s\) as the function that maps hybrid terms \(\tau \in T^{W}_{w, \ari}\) to \(\sigma(\tau)\),
every flexible operation symbol \(\sigma \colon \ari \to s\) as the function that maps hybrid terms \(\tau \in T^{W}_{w, \ari}\) to \(\sigma(w, \tau)\),
and every relation symbol (rigid or flexible) as the empty set.

\begin{lemma} [Hybrid-term model and its freeness]
  \label{lemma:htm-freeness}
  For every\/ \(\Sigma^{\nominal}\)-model\/ \(W\), \((W, T^{W})\) is a \(\Delta\)-model.
  Moreover, for any \(\Delta\)-model \((W', M')\) and first-order\/ \(\Sigma^{\nominal}\)-homomorphism \(f \colon W \to W'\) there exists a unique \(\Delta\)-homomorphism \(h \colon (W, T^{W}) \to (W', M')\) that agrees with \(f\) on \(W\).
\end{lemma}
\begin{proof}
  The fact that \((W, T^{W})\) is a \(\Delta\)-model is straightforward.
  For the `freeness' part, we define \(\{ h_{w} \colon T^{W}_{w} \to M'_{f(w)} \}_{w \in |W|}\) by structural induction on hybrid terms:
  \begin{enumerate}
  \item \(h_{w}(\sigma(\tau)) = M'_{f(w),\sigma}(h_{w}(\tau))\) for all \(w \in |W|\), \(\tau \in T^{W}_{w,\ari}\), and \(\sigma \in F^{\rigid}_{\ari \to s}\);

  \item \(h_{w}(\sigma(w, \tau)) = M'_{f(w),\sigma}(h_{w}(\tau))\) for all \(w \in |W|\), \(\tau \in T^{W}_{w,\ari}\), and \(\sigma \in F^{\flexible}_{\ari \to s}\);

  \item \(h_{w, s}(\tau) = h_{w_{0}, s}(\tau)\) for all \(w_{0}, w \in |W|\), \(\tau \in T^{W}_{w_{0}, s}\), and \(s \in S^{\rigid}\).

  \end{enumerate}
  It is again straightforward to check that \(h = (f, \{ h_{w} \}_{w \in |W|})\) is a Kripke homomorphism \((W, T^{W}) \to (W', M')\), and (also by induction) that it is unique with this property.
\end{proof}

\minisec{Standard term model}
When \(W\) is the first-order term model \(T_{\Sigma^{\nominal}}\), by Lemma~\ref{lemma:htm-freeness} we obtain the standard hybrid-term model over \(\Delta\), which we denote by \((T_{\Sigma^{\nominal}}, \{T^{\Delta}_{k}\}_{k \in T_{\Sigma^{\nominal}}})\).

The initiality of the standard term model provides a straightforward interpretation of hybrid terms in \(\Delta\)-models \((W, M)\): for every hybrid term \(t \in T^{\Delta}_{k}\), we denote by \((W, M)_{t}\) or \(M_{h(k), t}\) the image of \(t\) under \(h_{k}\), where \(h\) is the unique homomorphism \((T_{\Sigma^{\nominal}}, T^{\Delta}) \to (W, M)\).

\begin{proposition} [Reachability of the hybrid-term model]
  \label{proposition:tm-reachability}
  If\/ \(W\) is a \emph{reachable} \(\Sigma^{\nominal}\)-model, meaning that the first-order homomorphism \(T_{\Sigma^{\nominal}} \to W\) is surjective, then \((W, T^{W})\) is \emph{reachable} too; i.e.\ the unique homomorphism \(h \colon (T_{\Sigma^{\nominal}}, \{ T^{\Delta}_{k} \}_{k \in T_{\Sigma^{\nominal}}}) \to (W, \{ T^{W}_{w} \}_{w \in W})\) is surjective.
\end{proposition}
\begin{proof}
  By hypothesis, \(h \colon T_{\Sigma^{\nominal}} \to W\) is surjective.
  Therefore, all we need to prove is that, for every nominal term \(k\), the \(\Sigma\)-homomorphism \(h_{k} \colon T^{\Delta}_{k} \to T^{W}_{h(k)}\) is surjective.
  We proceed by induction on the structure of the hybrid terms over \(W\):
  \begin{enumerate}
  \item Assume that \(\tau \in T^{W}_{h(k),\ari}\) and \(\sigma \in F^{\rigid}_{\ari \to s}\).
    By the induction hypothesis, there exists a hybrid term \(t \in T^{\Delta}_{k,\ari}\) such that \(h_{k}(t) = \tau\).
    Therefore, \(h_{k}(\sigma(t)) = \sigma(h_{k}(t)) = \sigma(\tau)\).

  \item Assume that \(\tau \in T^{W}_{h(k),\ari}\) and \(\sigma \in F^{\flexible}_{\ari \to s}\).
    By the induction hypothesis, there exists \(t \in T^{\Delta}_{k,\ari}\) such that \(h_{k}(t) = \tau\).
    Therefore, \(h_{k}(\sigma(k, t)) = \sigma(h(k), h_{k}(t)) = \sigma(h(k), \tau)\).

  \item Assume that \(\tau \in T^{W}_{w_{0}, s}\) and \(s \in S^{\rigid}\).
    Since \(W\) is reachable, \(w_{0} = h(k_{0})\) for some nominal term \(k_{0}\).
    By the induction hypothesis, there exists \(t_{0} \in T^{\Delta}_{k_{0}, s}\) such that \(h_{k_{0}}(t) = \tau\).
    Finally, given that \(s\) is rigid, we have \(t \in T^{\Delta}_{k, s}\) and \(h_{k}(t) = h_{k_{0}}(t) = \tau\).
    \qedhere
  \end{enumerate}
\end{proof}

\minisec{Sentences}
The \emph{atomic sentences} \(\rho\) defined over a \(\HDFOLS\)-signature \(\Delta\) are given by:
\[
  \rho \Coloneqq
  k_{1} = k_{2} \alt
  \lambda(k') \alt
  t_{1} =_{k, s} t_{2} \alt
  \varpi(t) \alt
  \pi(k, t)
\]
where \(k, k_{i} \in T_{\Sigma^{\nominal}}\) are nominal terms,
\(k'\) is a tuple of terms corresponding to the arity of \(\lambda \in P^{\nominal}\),
\(t_{i} \in T^{\Delta}_{k, s}\) and \(t \in T^{\Delta}_{k,\ari}\) are (tuples of) hybrid terms,\footnote{Note that, if the arity \(\ari\) is rigid, then the sets \(\{T^{\Delta}_{k,\ari}\}_{k \in T_{\Sigma^{\nominal}}}\) coincide.}
\(\varpi \in P^{\rigid}_{\ari}\),
and \(\pi \in P^{\flexible}_{\ari}\).
We refer to these sentences, in order, as \emph{nominal equations}, \emph{nominal relations}, \emph{hybrid equations}, \emph{rigid hybrid relations}, and \emph{non-rigid/flexible hybrid relations}, respectively.
When there is no danger of confusion, we may drop one or both subscripts \(k, s\) from the notation \(t_{1} =_{k, s} t_{2}\).

\emph{Full sentences}\label{sentence-building-operators} over \(\Delta\) are built from atoms according to the following grammar:
\[
  \gamma \Coloneqq
  \rho \alt
  \act(k_{1}, k_{2}) \alt
  \at{k} \gamma \alt
  \lnot \gamma \alt
  \textstyle\bigand \Gamma \alt
  \store{z} \gamma' \alt
  \forAll{X} \gamma'' \alt
  \nec{\act} \gamma \alt
  \lnext{\sigma} \gamma
\]
where \(k, k_{i} \in T_{\Sigma^{\nominal}}\) are nominal terms,
\(\act \in A^{\nominal}\) is an action,
\(\Gamma\) is a finite set of sentences,
\(z\) is a nominal variable,
\(\gamma'\) is a sentence over the signature \(\Delta[z]\) obtained by adding \(z\) as a new constant to \(F^{\nominal}\),
\(X\) is a block of nominal and/or rigid variables,
\(\gamma''\) is a a sentence over the signature \(\Delta[X]\) obtained by adding the elements of \(X\) as a new constants to \(F^{\nominal}\) and \(F^{\rigid}\), and
\(\sigma\) is a unary operation symbol in \(F^{\nominal}\).
Other than the first two kinds of sentences (\emph{atoms} and \emph{action relations}), we refer to the sentence-building operators, in order, as \emph{retrieve}, \emph{negation}, \emph{conjunction}, \emph{store}, \emph{universal quantification}, \emph{necessity}, and \emph{next}, respectively.

We denote by \(\Sen^{\HDFOLS}(\Delta)\) the set of all \(\HDFOLS\)-sentences over \(\Delta\).

\minisec{The local satisfaction relation}
Given a Kripke structure \((W, M)\) for \(\Delta\) and a possible world \(w \in |W|\), we define the \emph{satisfaction of \(\Delta\)-sentences at \(w\)} by structural induction as follows:
\begin{enumerate}

\item \emph{For atomic sentences}:
  \begin{itemize}

  \item \((W, M) \models^{w} k_{1} = k_{2}\) iff \(W_{k_{1}} = W_{k_{2}}\) for all nominal equations \(k_{1} = k_{2}\);

  \item \((W, M) \models^{w} \lambda(k)\) iff \(W_{k} \in W_{\lambda}\) for all nominal relations \(\lambda(k)\);

  \item \((W, M) \models^{w} t_{1} =_{k} t_{2}\) iff \(M_{w', t_{1}} = M_{w', t_{2}}\), where \(w' = W_{k}\), for all equations \(t_{1} =_{k} t_{2}\);

  \item \((W, M) \models^{w} \varpi(t)\) iff \((W, M)_t \in M_{w,\varpi}\) for all rigid relations \(\varpi(t)\);

  \item \((W, M) \models^{w} \pi(k, t)\) iff \((W, M)_t \in M_{w',\pi}\), where \(w' = W_{k}\), for flexible relations \(\pi(k, t)\).

  \end{itemize}

\item \emph{For full sentences}:
  \begin{itemize}

  \item \((W, M) \models^{w} \act(k_{1}, k_{2})\) iff \((W_{k_{1}}, W_{k_{2}}) \in W_{\act}\) for all action relations \(\act(k_{1}, k_{2})\);

  \item \((W, M) \models^{w} \at{k} \gamma\) iff \((W, M) \models^{w'} \gamma\), where \(w' = W_{k}\);

  \item \((W, M) \models^{w} \neg \gamma\) iff \((W, M) \nmodels^{w} \gamma\);

  \item \((W, M) \models^{w} \textstyle\bigand \Gamma\) iff \((W, M) \models^{w} \gamma\) for all \(\gamma \in \Gamma\); 

  \item \((W, M) \models^{w} \store{z}{\gamma}\) iff \((W, M)^{z \assign w} \models^{w} \gamma\), \\
    where \(W^{z \assign w}\) is the unique \(\Delta[z]\)-expansion\footnote{In general, by a \(\Delta[X]\)-expansion of \((W, M)\) we understand a \(\Delta[X]\)-model \((W', M')\) that interprets all symbols in \(\Delta\) in the same way as \((W, M)\). In particular, \((W', M')\) has the same carrier sets as \((W, M)\).} of \(W\) that interprets the variable \(z\) as \(w\);

  \item \((W, M) \models^{w} \forAll{X}{\gamma}\) iff \((W', M') \models^{w} \gamma\) for all \(\Delta[X]\)-expansions \((W', M')\) of \((W, M)\);

  \item \((W, M) \models^{w} \nec{\act} \gamma\) iff \((W, M) \models^{w'}\gamma\) for all worlds \(w' \in |W|\) such that \((w, w') \in W_\act\);

  \item \((W, M) \models^{w} \lnext{\sigma} \gamma\) iff \((W, M) \models^{w'}\gamma\), where \(w' = W_{\sigma}(w)\).

  \end{itemize}

\end{enumerate}

\begin{fact}
  \label{fact:sat-atoms-actrel}
  The next properties follow easily from the definition of the satisfaction relation.
  \begin{enumerate}
    
  \item The satisfaction of atoms and of action relations \(\rho\) does not depend on the possible worlds.
    More specifically, \((W, M) \models^{w} \rho\) iff \((W, M) \models^{w'} \rho\) for all possible worlds \(w, w' \in |W|\).

  \item The satisfaction of atoms and of action relations \(\rho\) is preserved by homomorphisms.
    That is, if \((W, M) \models \rho\) and there is a homomorphism \((W, M) \to (W', M')\), then \((W', M') \models \rho\).
    
  \end{enumerate}
\end{fact}

To state the \emph{satisfaction condition} -- and thus finalize the presentation of \(\HDFOLS\) -- let us first notice that every signature morphism \(\varphi \colon \Delta \to \Delta'\) induces appropriate \emph{translations of sentences} and \emph{reductions of models}, as follows:
every \(\Delta\)-sentence \(\gamma\) is translated to a \(\Delta'\)-sentence \(\varphi(\gamma)\) by replacing (usually in an inductive manner) the symbols in \(\Delta\) with symbols from \(\Delta'\) according to \(\varphi\);
and every \(\Delta'\)-model \((W', M')\) is reduced to a \(\Delta\)-model \((W', M') \red_{\varphi}\) that interprets every symbol \(x\) in \(\Delta\) as \((W', M')_{\varphi(x)}\).
When \(\varphi\) is an inclusion, we usually denote \((W', M') \red_{\varphi}\) by \((W', M') \red_{\Delta}\) -- in this case, the model reduct simply forgets the interpretation of those symbols in \(\Delta'\) that do not belong to \(\Delta\).

The following satisfaction condition can be proved by induction on the structure of \(\Delta\)-sentences.
Its argument is virtually identical with several other variants presented in the literature (see, e.g.~\cite{Diaconescu16}), hence we choose to present the result without a proof.

\begin{proposition}[Local satisfaction condition for signature morphisms]
  \label{proposition:morph-sat-cond}
  For every signature morphism \(\varphi \colon \Delta \to \Delta'\), \(\Delta'\)-model \((W', M')\), world \(w' \in |W'|\), and \(\Delta\)-sentence \(\gamma\), we have:
  \stepcounter{footnote}
  \footnotetext{Note that, by the definition of model reducts, \((W', M')\) and \((W', M') \red_{\varphi}\) have the same possible worlds.}
  \addtocounter{footnote}{-1}
  \[
    (W', M') \models^{w'} \varphi(\gamma)
    \quad\text{if and only if}\quad
    (W', M') \red_{\varphi} \models^{w'} \gamma.\footnotemark
    \pushQED{\qed}
    \qedhere
  \]
\end{proposition}

\minisec{Substitutions}
Consider two signature extensions \(\Delta[X]\) and \(\Delta[Y]\) with blocks of variables, and let \(X = X^{\nominal} \cup X^{\rigid}\) and \(Y = Y^{\nominal} \cup Y^{\rigid}\) be the partitions of \(X\) and \(Y\) into blocks of nominal and rigid variables.
A \emph{\(\Delta\)-substitution} \(\theta \colon X \to Y\) consists of a pair of functions
\(\theta^{\nominal} \colon X^{\nominal} \to T_{\Sigma^{\nominal}[Y^{\nominal}]}\) and
\(\theta^{\rigid} \colon X^{\rigid} \to T^{\Delta[Y]}_{k}\), where \(k\) is a nominal term.\footnote{Since the sorts of the variables are rigid, it does not matter which nominal term \(k\) we choose.}

Similarly to signature morphisms, \(\Delta\)-substitutions \(\theta \colon X \to Y\) determine translations of \(\Delta[X]\)-sentences into \(\Delta[Y]\)-sentences, and reductions of \(\Delta[Y]\)-models to \(\Delta[X]\)-models.
The proofs of the next two propositions are similar to the ones given in~\cite{Gaina17Her} for hybrid substitutions.

\begin{proposition}[Local satisfaction condition for substitutions]
  \label{proposition:subst-sat-cond}
  For every \(\Delta\)-substitution \(\theta \colon X \to Y\), \(\Delta[Y]\)-model \((W, M)\), world \(w \in |W|\), and \(\Delta[X]\)-sentence \(\gamma\), we have:
  \[
    (W, M) \models^{w} \theta(\gamma)
    \quad\text{if and only if}\quad
    (W, M) \red_{\theta} \models^{w} \gamma.
    \pushQED{\qed}\qedhere
  \]
\end{proposition} 

Proposition~\ref{proposition:subst-sat-cond}, together with Proposition~\ref{proposition:rm-gen-subst} below, has an important technical role in the initial-model and Herbrand's-theorem proofs presented in the later sections of the paper.

\begin{proposition} [Substitutions generated by expansions of reachable models]
  \label{proposition:rm-gen-subst}
  If \((W, M)\) is a reachable \(\Delta\)-model, then for every \(\Delta[X]\)-expansion \((W', M')\) of \((W, M)\) there exists a substitution \(\theta \colon X \to \emptyset\) such that \((W, M) \red_{\theta} = (W', M')\).
  \pushQED{\qed}\qedhere
\end{proposition}

\minisec{Expressive power}
Fact~\ref{fact:sat-atoms-actrel} highlights one of the main distinguishing features of \(\HDFOLS\):
the satisfaction of atomic sentences, whether they involve flexible symbols or not, does not depend on the possible world where the sentences are evaluated.
This contrasts the standard approach in hybrid logic, where each nominal is regarded as an atomic sentence satisfied precisely at the world that corresponds to the interpretation of that nominal.
In \(\HDFOLS\), the dependence of the satisfaction of sentences on possible worlds is explicit rather than implicit, and is achieved through the \emph{store} operator.
These two approaches are often inter-definable:
every nominal sentence \(k\) is semantically equivalent to \(\store{z} z = k\);\footnote{Formally, in hybrid logic, \((W, M) \models^{w} k\) if and only if \(w = W_{k}\).}
moreover, when the logical system admits existential quantification over nominal variables, every \emph{store} sentence \(\store{z} \gamma\) is semantically equivalent to \(\Exists{z} (z \land \gamma)\).
Following the lines of~\cite[Section~4.3]{Gaina17Her}, the first part of this correspondence can be used to show that even without considering action relations, \(\HDFOLS\) is strictly more expressive than other standard hybrid logics constructed from the same base logic such as the hybrid first-order logic with rigid symbols~\cite{DiaconescuM16,Diaconescu16}.

As it is often the case, other sentence-building operators can be derived from those presented on page~\pageref{sentence-building-operators} through standard constructions: for instance, \emph{implication} can be defined based on conjunction and negation, \emph{existential quantification} based on universal quantification and negation, and the modal-logic \emph{possibility} operator based on necessity and negation.

\emph{Necessity} and \emph{next} can be derived as well, since all sentences \(\nec{\act} \gamma\) and \(\lnext{\sigma} \gamma\) are semantically equivalent to \(\store{z} \forAll{z'} \act(z, z') \implies \at{z'} \gamma\) and \(\store{z} \at{\sigma(z)} \gamma\), respectively.
However, we have chosen to present them independently, as there are logical systems that do not support the \emph{store} operator or \emph{universal quantifiers}.
The results in this paper are modular; their proofs are not based on semantic abbreviations such as those listed above, nor on the existence of certain operators -- with one exception: the main theorems rely on the existence of \emph{retrieve}.

%%%%%%%%%%%%%%%%%%%%%%%%%%%%%%%%%%%%%%%%%%%%%%%%%%%%%%%%%%%%%%%%%%%%%%%%%%%%%%%% 

\section{Initiality}
\label{section:initiality}

Hybrid-dynamic \emph{Horn clauses} are obtained from atomic sentences by repeated applications of%
\begin{inparenum}
  \inparitem retrieve
  \inparitem implication such that the condition is a conjunction of atomic sentences or action relations, 
  \inparitem store,
  \inparitem universal quantification,
  \inparitem necessity, or
  \inparitem next.
\end{inparenum}
We denote by \(\HDCLS\) the Horn-clause fragment of \(\HDFOLS\) -- i.e.\ with the same signatures and models as \(\HDFOLS\), and with sentences defined as Horn clauses.
This restriction is crucial, because it entails that any set of sentences in \(\HDCLS\) has an initial model, as we prove next.

%%%%%%%%%%%%%%%%%%%%%%%%%%%%%%%%%%%%%%%%%%%%%%%%%%%%%%%%%%%%%%%%%%%%%%%%%%%%%%%% 

\subsection{The basic level}
\label{subsection:basic-level}

We first deal with the initiality property for the basic level of atoms.
The following result can also be found in~\cite{Gaina17Her}, though in a simplified setting where \(\Sigma^{\nominal}\) consists only of constants.

\begin{proposition}[Initiality of nominal equations and relations]
  \label{proposition:HDCLSn-initiality}
  Every set\/ \(\Gamma_{\nominal}\) of nominal equations or relations admits a reachable initial model \((W^{\Gamma_{\nominal}}, M^{\Gamma_{\nominal}})\).
\end{proposition}
\begin{proof}
  Consider the following relation on the first-order term model \(T_{\Sigma^{\nominal}}\):
  \(k_{1} \equiv k_{2}\) if and only if \(\Gamma_{\nominal} \models k_{1} = k_{2}\), for any two nominal terms \(k_{1}\) and \(k_{2}\) over \(\Delta\).

  It is straightforward to check that \(\equiv\) is \(\Sigma^{\nominal}\)-congruence on \(T_{\Sigma^{\nominal}}\).
  Therefore, we can define \(W^{\Gamma_{\nominal}}\) as the \(\Sigma^{\nominal}\)-model obtained from \(T_{\Sigma^{\nominal}} /{\equiv}\) by interpreting each relation symbol \(\lambda \in P^{\nominal}\) as the set \(\{ [k] \in T_{\Sigma^{\nominal}} /{\equiv} \mid \Gamma_{\nominal} \models \lambda(k) \}\), where \([k]\) denotes the congruence class of \(k\) modulo \(\equiv\).

  Now let \((W^{\Gamma_{\nominal}}, M^{\Gamma_{\nominal}})\) be the (reachable) hybrid-term model defined over \(W^{\Gamma_{\nominal}}\).
  That is, the possible worlds of \((W^{\Gamma_{\nominal}}, M^{\Gamma_{\nominal}})\) are congruence classes of nominal terms modulo \(\equiv\), and for each nominal term \(k\), the elements of \(M^{\Gamma_{\nominal}}_{[k]}\) are obtained from standard hybrid terms (as per the definition in Section~\ref{section:HDFOLS}) by replacing each nominal with its congruence class.
  For notational convenience, we denote the interpretation of a term \(t\) in \((W^{\Gamma_{\nominal}}, M^{\Gamma_{\nominal}})\) by \([t]\).

  All is left to prove is the universal property of \((W^{\Gamma_{\nominal}}, M^{\Gamma_{\nominal}})\).
  First, notice that, by construction, we have \((W^{\Gamma_{\nominal}}, M^{\Gamma_{\nominal}}) \models \Gamma_{\nominal}\).
  Then, for every Kripke model \((W, M)\) that satisfies \(\Gamma_{\nominal}\), there exists a unique homomorphism \(h \colon (W^{\Gamma_{\nominal}}, M^{\Gamma_{\nominal}}) \to (W, M)\), where:
  \begin{itemize}
  \item on possible worlds, \(h \colon W^{\Gamma_{\nominal}} \to W\) is the unique \(\Sigma^{\nominal}\)-homomorphism between \(W^{\Gamma_{\nominal}}\) and \(W\), defined by \(h([k]) = W_{k}\) for all \(k \in T_{\Sigma^{\nominal}}\), and 

  \item for every \(k \in T_{\Sigma^{\nominal}}\), \(h_{[k]} \colon M^{\Gamma_{\nominal}}_{[k]} \to M_{W_{k}}\) is defined by \(h_{[k]}([t]) = (W, M)_{t}\) for all \(t \in T^{\Delta}_{k}\).
  \end{itemize}
  It is straightforward to check that the definition of \(h\) is consistent, since \((W, M) \models \Gamma_{\nominal}\).
\end{proof}

In order to extend Proposition~\ref{proposition:HDCLSn-initiality} to arbitrary sets of \(\HDCLS\)-sentences, we use a notion of congruence on a Kripke structure and the universal property of its corresponding quotient.

\begin{definition}[Congruence]
  \label{definition:congruence}
  Let \((W, M)\) be a Kripke structure over a \(\HDFOLS\)-signature \(\Delta\).
  A \emph{\(\Delta\)-congruence} \({\equiv} = \{\equiv_{w}\}_{w \in |W|}\) on \((W, M)\) is a family of \(\Sigma\)-congruences \(\equiv_{w}\) on \(M_{w}\), for each \(w \in |W|\), such that \((\equiv_{w_{1}, s}) = (\equiv_{w_{2}, s})\) for all \(w_{1}, w_{2} \in |W|\) and \(s \in S^{\rigid}\).
\end{definition}

The following construction and universal property are straightforward generalizations of their first-order counterparts, and have been studied in several other papers in the literature (see, e.g.~\cite{Gaina17Bir}).
For that reason, we include them for further reference without a proof.

\begin{proposition}[Quotient model]
  \label{proposition:quotient-model}
  Every \(\HDCLS\)-congruence \(\equiv\) on a \(\Delta\)-model \((W, M)\) determines a \emph{quotient-model homomorphism} \((\_ /{\equiv}) \colon (W, M) \to (W, M/{\equiv}) \) that acts as an identity on possible worlds, and for which \((M/{\equiv})_{w}\) is the quotient \(\Sigma\)-model \(M_{w}/{\equiv_{w}}\).
  
  Moreover, \((\_ / {\equiv})\) has the following \emph{universal property}:
  for any Kripke homomorphism \(h \colon (W, M) \to (W', M')\) such that \({\equiv} \subseteq \ker(h)\),\footnote{This means that \(h_{w, s}(a_{1}) = h_{w, s}(a_{2})\) for all \(a_{1}, a_{2} \in M_{w, s}\) such that \(a_{1} \equiv_{w, s} a_{2}\).}
  there exists a unique homomorphism \(h' \colon (W , M /{\equiv}) \to (W', M')\)
  such that \((\_/{\equiv}) \comp h' = h\).\footnote{We use the diagrammatic notation for function composition, where \(((\_/{\equiv}) \comp h')(a) = h'(a/{\equiv})\).}
  % for which the diagram below commutes.
  % \[
  %   \xymatrix @!0 @R=3em @C=5em {
  %   {(W, M)}
  %   \ar [rr] ^{(\_ /{\equiv})}
  %   \ar [dr] _{h}
  %   &
  %   & {(W, M /{\equiv})}
  %   \ar [dl] ^{h'}
  %   \\
  %   & {(W', M')}
  % }
  % \]
  \qed
\end{proposition}

\begin{proposition}
  \label{proposition:Gamma-congruence}
  Let\/ \(\Gamma\) be a set of nominal or hybrid equations over \(\Delta\), and\/
  \(W\) a reachable \(\Sigma^{\nominal}\)-model such that\/
  \(\Gamma \models \Gamma_{W}\),\footnote{By the definition of \emph{semantic entailment}, this means that all models satisfying \(\Gamma\) satisfy \(\Gamma_{W}\) as well.} where \(\Gamma_{W} = \{k_{1} = k_{2} \in \Sen(\Sigma^{\nominal}) \mid W_{k_{1}} = W_{k_{2}}\}\).
  Then \(\Gamma\) generates a congruence \(\equiv\) on \((W, T^{W})\) defined by
  \(\tau_{1} \equiv_{w} \tau_{2}\) iff there exist a nominal term \(k\) and hybrid terms \(t_{1}\) and \(t_{2}\) such that\/
  \(\Gamma \models t_{1} =_{k} t_{2}\), \(w = W_{k}\), and \(\tau_{i} = T^{W}_{w, t_{i}}\).
\end{proposition}

\begin{proof}
  We first show that for every model \((W', M')\) of \(\Gamma\) there exists a homomorphism \((W, T^{W}) \to (W', M')\).
  Since \(W\) is reachable, the unique homomorphism \(q \colon T_{\Sigma^{\nominal}} \to W\) is surjective.
  If \((W', M') \models \Gamma\), then \(W' \models \Gamma_{W}\).
  Therefore, the unique homomorphism \(f \colon T_{\Sigma^{\nominal}} \to W'\) satisfies \(\ker(q) \subseteq \ker(f)\).
  By the universal property of first-order model quotients, there exists a morphism \(h \colon W \to W'\) such that \(q \comp h = f\),
  which can then be uniquely extended to a Kripke homomorphism \(h \colon (W, T^{W}) \to (W', M')\) as per Lemma~\ref{lemma:htm-freeness}.

  The above observation allows us to prove that the definition of \(\equiv_{w}\) does not depend on the choice of \(k\).
  That is, if \(w = W_{k} = W_{k'}\) and \(\Gamma \models t_{1} =_{k} t_{2}\),
  then there exist \(t'_{1}, t'_{2} \in T^{\Delta}_{k'}\) such that \(T^{W}_{w, t_{i}} = T^{W}_{w, t'_{i}}\) and \(\Gamma \models t'_{1} =_{k'} t'_{2}\).
  It suffices to consider the following list of inferences:
  \begin{proofsteps}{22.5em}
    \(T^{W}_{w, t_{i}} = T^{W}_{w, t'_{i}}\) for some \(t'_{1}, t'_{2} \in T^{\Delta}_{k'}\)
    & by Proposition~\ref{proposition:tm-reachability}
    \\
    \label{prop:cong0-1}%
    \((W', M')_{t_{i}} = h(T^{W}_{w, t_{i}}) = h(T^{W}_{w, t'_{i}}) = (W', M')_{t'_{i}}\)
    \newline for all \(\Delta\)-models \((W', M')\) such that \((W', M') \models \Gamma\)
    & since there exists a morphism \(h \colon (W, T^{W}) \to (W', M')\)
    \\
    \label{prop:cong0-2}%
    \((W', M')_{t_{1}} = (W', M')_{t_{2}}\)
    \newline for all \(\Delta\)-models \((W', M')\) such that \((W', M') \models \Gamma\)
    & since \(\Gamma \models t_{1} =_{k} t_{2}\)
    \\
    \((W', M')_{t'_{1}} = (W', M')_{t'_{2}}\)
    \newline for all \(\Delta\)-models \((W', M')\) such that \((W', M') \models \Gamma\)
    & from \ref{prop:cong0-1} and \ref{prop:cong0-2}
    \\
    \(\Gamma \models t'_{1} =_{k'} t'_{2}\)
    & by the definition of \(\models\)
  \end{proofsteps}

  The fact that \(\equiv_{w}\) is both reflexive and symmetric is a straightforward consequence of the reachability of \(T^{W}\).
  Therefore, we focus on the transitivity and the compatibility of \(\equiv_{w}\) with the operations in \(F\).
  For transitivity, suppose \(\tau_{1} \equiv_{w} \tau_{2}\) and \(\tau_{2} \equiv_{w} \tau_{3}\).
  Then:
  \begin{proofsteps}{22.5em}
    \label{prop:cong1-1}%
    \(\Gamma \models t_{1} =_{k} t_{2}\), \(w = W_{k}\), and \(\tau_{i} = T^{W}_{w, t_{i}}\) for \(i \in \{1, 2\}\),
    \newline for some nominal term \(k\) and hybrid terms \(t_{1}, t_{2} \in T^{\Delta}_{k}\)
    & by the definition of \(\equiv_{w}\)
    \\
    \label{prop:cong1-2}%
    \(\Gamma \models t'_{2} =_{k} t'_{3}\), \(w = W_{k}\), and \(\tau_{i} = T^{W}_{w, t'_{i}}\) for \(i \in \{2, 3\}\),
    \newline for some nominal term \(k\) and hybrid terms \(t'_{2}, t'_{3} \in T^{\Delta}_{k}\)
    & by the definition of \(\equiv_{w}\)
    \\
    \((W', M')_{t_{2}} = h(T^{W}_{w, t_{2}}) = h(T^{W}_{w, t'_{2}}) = (W', M')_{t'_{2}}\)
    \newline for all \(\Delta\)-models \((W', M')\) such that \((W', M') \models \Gamma\)
    & since there exists a morphism \(h \colon (W, T^{W}) \to (W', M')\)
    \\
    \label{prop:cong1-3}%
    \(\Gamma \models t_{2} =_{k} t'_{2}\)
    & by the definition of \(\models\)
    \\
    \(\Gamma \models t_{1} =_{k} t'_{3}\)
    & from \ref{prop:cong1-1}, \ref{prop:cong1-3}, and \ref{prop:cong1-2}
    \\
    \(\tau_{1} \equiv_{w} \tau_{3}\)
    & by the definition of \(\equiv_{w}\)
  \end{proofsteps}

  For the compatibility of \(\equiv_{w}\) with the operations in \(F\), assume \(\sigma \in F_{\ari \to s}\) and \(\tau_{1}, \tau_{2} \in T^{W}_{w, \ari}\) such that \(\tau_{1} \equiv_{w} \tau_{2}\).
  There is virtually no distinction between the case where \(\sigma\) is rigid and the case where \(\sigma\) is flexible, hence it suffices to make explicit the proof for \(\sigma \in F^{\rigid}_{\ari \to s}\):
  \begin{proofsteps}{22.5em}
    \label{prop:cong1-1}%
    \(\Gamma \models t_{1} =_{k} t_{2}\), \(w = W_{k}\), and \(\tau_{i} = T^{W}_{w, t_{i}}\) for \(i \in \{1, 2\}\),
    \newline for some nominal \(k\) and hybrid terms \(t_{1}, t_{2} \in T^{\Delta}_{k}\)
    & by the definition of \(\equiv_{w}\)
    \\
    \(\Gamma \models \sigma(t_{1}) =_{k} \sigma(t_{2})\)
    & by the general properties of \(\models\)
    \\
    \(T^{W}_{w,\sigma(t_{1})} \equiv_{w} T^{W}_{w,\sigma(t_{2})}\)
    & by the definition of \(\equiv_{w}\)
    \\
    \((T^{W}_{w})_{\sigma}(\tau_{1}) \equiv_{w} (T^{W}_{w})_{\sigma}(\tau_{2})\)
    & since \(T^{W}_{w,\sigma(t_{i})} = (T^{W}_{w})_{\sigma}(\tau_{i})\)
  \end{proofsteps}

  Finally, we need to show that the relations \(\{\equiv_{w}\}_{w \in |W|}\) coincide on rigid sorts, that is
  \((\equiv_{w, s}) = (\equiv_{w', s})\) for all \(w, w' \in |W|\) and \(s \in S^{\rigid}\).
  Suppose \(\tau_{1} \equiv_{w, s} \tau_{2}\).
  It follows that:
  \begin{proofsteps}{22.5em}
    \(\Gamma \models t_{1} =_{k, s} t_{2}\), \(w = W_{k}\), and \(\tau_{i} = T^{W}_{w, t_{i}}\) for \(i \in \{1, 2\}\),
    \newline for some nominal \(k\) and hybrid terms \(t_{1}, t_{2} \in T^{\Delta}_{k, s}\)
    & by the definition of \(\equiv_{w, s}\)
    \\
    \(w' = W_{k'}\) for some nominal \(k'\)
    & since \(W\) is reachable
    \\
    \(\Gamma \models t_{1} =_{k', s} t_{2}\) and \(t_{1}, t_{2} \in T^{\Delta}_{k', s}\)
    & since \(s\) is rigid
    \\
    \(\tau_{1} \equiv_{w', s} \tau_{2}\)
    & by the definition of \(\equiv_{w', s}\)
    \qedhere
  \end{proofsteps}
\end{proof}

Unlike first-order congruences and the \(\Delta\)-congruences presented in Definition~\ref{definition:congruence}, the general notion of Kripke congruence is significantly more complex when it involves the formation of a quotient model for the possible worlds.
In this paper, we consider congruences over hybrid-term models, for which the quotients can be obtained in two stages: first on possible worlds, then on their local models.
This constitutes a step forward in understanding Kripke congruences and their implications.
Therefore, Proposition~\ref{proposition:Gamma-congruence} is an interesting result in its own right, since it deals with the congruence generated by a set of equations.
We use it next as a basis for proving the existence of initial models for atomic sentences.

\begin{theorem} [Atomic initiality]
  \label{theorem:HDCLS0-initiality}
  Every set\/ \(\Gamma\) of atomic sentences over a \(\HDCLS\)-signature \(\Delta = (\Sigma^{\nominal}, \Sigma^{\rigid} \subseteq \Sigma)\) admits a reachable initial model \((W^{\Gamma}, M^{\Gamma})\).
\end{theorem}
\begin{proof}
  Let \(\Gamma_{\nominal}\) be the subset of nominal equations and relations in \(\Gamma\).
  By Proposition~\ref{proposition:HDCLSn-initiality}, there exists a reachable initial model \((W^{\Gamma_{\nominal}}, M^{\Gamma_{\nominal}})\) of \(\Gamma_{\nominal}\).
  For notational convenience, we denote the interpretation of a (nominal or hybrid) term \(t\) in \((W^{\Gamma_{\nominal}}, M^{\Gamma_{\nominal}})\) by \([t]\).

  For each \(w \in |W^{\Gamma_{\nominal}}|\), we define a relation \(\equiv_{w}\) on \(M^{\Gamma_{\nominal}}_{w}\) as follows:
  \(\tau_{1} \equiv_{w} \tau_{2}\) iff there exist a nominal \(k\) and hybrid terms \(t_{1}, t_{2} \in T^{\Delta}_{k}\) such that \(w = [k]\), \(\tau_{i} = [t_{i}]\), and \(\Gamma \models t_{1} =_{k} t_{2}\).
  Since \(W^{\Gamma_\nominal}\) is reachable, by Proposition~\ref{proposition:Gamma-congruence}, \(\equiv\) is a \(\Delta\)-congruence on \((W^{\Gamma_{\nominal}}, M^{\Gamma_{\nominal}})\).

  By Proposition~\ref{proposition:quotient-model}, the congruence \(\equiv\) determines a quotient-model homomorphism \((\_ /{\equiv})\) between  \((W^{\Gamma_{\nominal}}, M^{\Gamma_{\nominal}})\) and \((W^{\Gamma_{\nominal}}, M^{\Gamma_{\nominal}}/{\equiv})\).
  Let \((W^{\Gamma}, M^{\Gamma})\) be the \(\Delta\)-model obtained from \((W^{\Gamma_{\nominal}}, M^{\Gamma_{\nominal}}/{\equiv})\) by interpreting, for every \(w \in |W^{\Gamma_{\nominal}}|\), each relation symbol
  \begin{itemize}
  \item \(\pi \in P^{\rigid}\) as \(M^{\Gamma}_{w, \pi} = \{[t]/{\equiv_{w}} \in M^{\Gamma_{\nominal}}_{w}/{\equiv_{w}} \mid t \in T^{\Delta}_{k}, [k] = w,\ \text{and}\ \Gamma \models \pi(t)\}\), and
  \item \(\pi \in P^{\flexible}\) as \(M^{\Gamma}_{w, \pi} = \{[t]/{\equiv_{w}} \in M^{\Gamma_{\nominal}}_{w}/{\equiv_{w}} \mid t \in T^{\Delta}_{k}, [k] = w,\ \text{and}\ \Gamma \models \pi(k, t)\}\).
  \end{itemize}
  Trivially, since \((W^{\Gamma_{\nominal}}, M^{\Gamma_{\nominal}})\) is reachable and \((\_ /{\equiv})\) is surjective, then \((W^{\Gamma}, M^{\Gamma})\) is reachable as well.
  Therefore, all we need to prove is that \((W^{\Gamma}, M^{\Gamma})\) is an initial model of \(\Gamma\).

  The fact that \((W^{\Gamma}, M^{\Gamma})\) is a model of \(\Gamma\) follows in a straightforward manner from the very construction of \((W^{\Gamma}, M^{\Gamma})\).
  We focus on the initiality property.
  Let \((W, M)\) be a \(\Delta\)-model satisfying \(\Gamma\).
  Since \(\Gamma_{\nominal} \subseteq \Gamma\), we have \((W, M) \models \Gamma_{\nominal}\), and thus, by Proposition~\ref{proposition:HDCLSn-initiality}, there exists a unique \(h_{\nominal} \colon (W^{\Gamma_{\nominal}}, M^{\Gamma_{\nominal}}) \to (W, M)\).
  Moreover, for every \(\tau_{1} \equiv_{w} \tau_{2}\) we have:
  \begin{proofsteps}{22em}
    \(\Gamma \models t_{1} =_{k} t_{2}\), \(w = [k]\), and \(\tau_{i} = [t_{i}]\) for \(i \in \{1, 2\}\),
    \newline for some nominal term \(k\) and hybrid terms \(t_{1}, t_{2} \in T^{\Delta}_{k}\)
    & by the definition of \(\equiv_{w}\)
    \\
    \((W, M) \models t_{1} =_{k} t_{2}\)
    & since \((W, M) \models \Gamma\)
    \\
    \(h_{\nominal}(\tau_{1}) = h_{\nominal}(\tau_{2})\)
    & since \(h_{\nominal}(\tau_{i}) = h_{\nominal}([t_{i}]) = (W, M)_{t_{i}}\)
  \end{proofsteps}
  This shows that \({\equiv} \subseteq \ker(h_{\nominal})\).
  By Proposition~\ref{proposition:quotient-model}, there exists a unique Kripke homomorphism \(h \colon (W^{\Gamma_{\nominal}}, M^{\Gamma_{\nominal}}/{\equiv}) \to (W, M)\) such that \((\_/{\equiv}) \comp h = h_{\nominal}\).
  To finalize the proof, we need to show that \(h\) preserves the interpretation of the relation symbols in \(\Delta\).

  We consider only the case where \(\pi\) is a rigid relation symbol.
  Non-rigid relation symbols can be treated in a similar manner.
  Assume \(a \in M^{\Gamma}_{w, \pi}\) for some \(w \in |W^{\Gamma}|\).
  It follows that:
  \begin{proofsteps}{22em}
    \(\Gamma \models \pi(t)\), \(w = [k]\), and \(a = [t]\)
    \newline for some nominal term \(k\) and hybrid term \(t \in T^{\Delta}_{k}\)
    & by the definition of \(M^{\Gamma}_{w, \pi}\)
    \\
    \(M_{W_{k}, t} \in M_{W_{k}, \pi}\)
    & since \((W, M) \models \Gamma\)
    \\
    \(M_{h(w), t} \in M_{h(w), \pi}\)
    & since \(h(w) = h([k]) = W_{k}\)
    \\
    \(h_{w}(a) \in M_{h(w), \pi}\)
    & since \(h_{w}(a) = h_{w}([t]) = M_{h(w), t}\)
    \qedhere
  \end{proofsteps}
\end{proof}

\begin{corollary}
  \label{corollary:HDCLS0-initiality}
  Under the notations and hypotheses of Theorem~\ref{theorem:HDCLS0-initiality}, we have that\/
  \(\Gamma \models \rho\) if and only if \((W^{\Gamma}, M^{\Gamma}) \models \rho\),
  for all atomic sentences or action relations \(\rho\) over \(\Delta\).
\end{corollary}
\begin{proof}
  Since \((W^{\Gamma}, M^{\Gamma})\) is a model of \(\Gamma\), the `only if' part holds trivially by the definition of the semantic entailment.
  For the `if' part, let \(\rho\) be an atomic sentence or an action relation over \(\Delta\) such that \((W^{\Gamma}, M^{\Gamma}) \models \rho\), and \((W, M)\) an arbitrary \(\Delta\)-model that satisfies \(\Gamma\).
  Since \((W^{\Gamma}, M^{\Gamma})\) is an initial model of \(\Gamma\), there exists a unique homomorphism \((W^{\Gamma}, M^{\Gamma}) \to (W, M)\).
  And because the satisfaction of both atomic sentences and action relations is preserved by homomorphisms (by Fact~\ref{fact:sat-atoms-actrel}), it follows that \((W, M) \models \rho\).
  Therefore, \(\Gamma \models \rho\).
\end{proof}

%%%%%%%%%%%%%%%%%%%%%%%%%%%%%%%%%%%%%%%%%%%%%%%%%%%%%%%%%%%%%%%%%%%%%%%%%%%%%%%% 

\subsection{The upper level}
\label{subsection:upper-level}

In this subsection, we lift the initiality property from the basic level to the level of Horn clauses.
To that end, let us denote by \(\HDFOLS_{0}\) the atomic fragment of \(\HDFOLS\).

\begin{proposition}
  \label{proposition:init}
  Let\/ \(\Gamma\) be a set of Horn clauses over \(\Delta\) and \((W^{\Gamma_{0}}, M^{\Gamma_{0}})\) an initial model of the set\/
  \(\Gamma_{0} = \{\rho \in \Sen^{\HDFOLS_{0}}(\Delta) \mid \Gamma \models \rho\}\).
  For each nominal \(k\) and clause \(\gamma\) over \(\Delta\),
  \[
    \Gamma \models \at{k} \gamma
    \quad\text{implies}\quad
    (W^{\Gamma_{0}}, M^{\Gamma_{0}}) \models \at{k} \gamma.
  \]
\end{proposition}
\begin{proof}
  We prove the result by induction on the structure of the clauses \(\gamma\) built over \(\Delta\).
  \begin{proofcases}

  \item[\(\gamma = \rho \in \Sen^{\HDFOLS_{0}}(\Delta)\)]
    If \(\Gamma \models \at{k} \rho\), then we obtain:
    \begin{proofsteps}{18em}
      \(\at{k} \rho \bimodels \rho\)
      & since \(\rho\) is atomic
      \\
      \(\Gamma \models \rho\)
      & since \(\Gamma \models \at{k} \rho \bimodels \rho\)
      \\
      \(\rho \in \Gamma_{0}\)
      & by the definition of \(\Gamma_{0}\)
      \\
      \((W^{\Gamma_{0}}, M^{\Gamma_{0}}) \models \rho\)
      & by Corollary~\ref{corollary:HDCLS0-initiality}
      \\
      \((W^{\Gamma_{0}}, M^{\Gamma_{0}}) \models \at{k} \rho\)
      & since \(\at{k} \rho \bimodels \rho\)
    \end{proofsteps}

  \item[\(\at{k'} \gamma\)]
    This case is trivial, because \(\at{k} \at{k'} \gamma \bimodels \at{k'} \gamma\).
    
  \item[\(\bigand H \implies \gamma\)]
    Suppose \(\Gamma \models \at{k} (\bigand H \implies \gamma)\) and \((W^{\Gamma_{0}}, M^{\Gamma_{0}}) \models^{w} H\), where \(w = W^{\Gamma_{0}}_{k}\).
    Then:
    \begin{proofsteps}{18em}
      \label{proposition:init-imp-pc1}%
      \((W^{\Gamma_{0}}, M^{\Gamma_{0}}) \models H\)
      & by Fact~\ref{fact:sat-atoms-actrel}
      \\
      \(\Gamma_{0} \models H\)
      & by Corollary~\ref{corollary:HDCLS0-initiality}
      \\
      \(\Gamma \models H\)
      & since \(\Gamma \models \Gamma_{0}\)
      \\
      \(\Gamma \models \at{k} \gamma\)
      & since \(H \cup \{\at{k} (\bigand H \implies \gamma)\} \models \at{k} \gamma\)
      \\
      \((W^{\Gamma_{0}}, M^{\Gamma_{0}}) \models \at{k} \gamma\)
      & by the induction hypothesis
      \\
      \label{proposition:init-imp-pc2}%
      \((W^{\Gamma_{0}}, M^{\Gamma_{0}}) \models^{w} \gamma\)
      & since \(w = W^{\Gamma_{0}}_{k}\)
      \\
      \((W^{\Gamma_{0}}, M^{\Gamma_{0}}) \models^{w} \bigand H \implies \gamma\)
      & from \ref{proposition:init-imp-pc1}--\ref{proposition:init-imp-pc2}, by the definition of \(\models\)
      \\
      \((W^{\Gamma_{0}}, M^{\Gamma_{0}}) \models \at{k} (\bigand H \implies \gamma)\)
      & by the definition of \(\models\)
    \end{proofsteps}

  \item[\(\store{z} \gamma\)]
    This case is straightforward, because \(\at{k} \store{z} \gamma \bimodels \at{k} \theta_{z \assign k}(\gamma)\), where \(\theta_{z \assign k} \colon \{z\} \to \emptyset\) is the \(\Delta\)-substitution that maps the nominal variable \(z\) to the (ground) term \(k\).
    % defined by \(\theta_{z \assign k}(z) = k\).

  \item[\(\forAll{X} \gamma\)]
    Suppose \(\Gamma \models \at{k} \forAll{X} \gamma\) and let \((W, M)\) be a \(\Delta[X]\)-expansion of \((W^{\Gamma_{0}}, M^{\Gamma_{0}})\).
    Since \((W^{\Gamma_{0}}, M^{\Gamma_{0}})\) is reachable, by Proposition~\ref{proposition:rm-gen-subst}, there exists a \(\Delta\)-substitution \(\theta \colon X \to \emptyset\) such that \((W^{\Gamma_{0}}, M^{\Gamma_{0}}) \red_{\theta} = (W, M)\).
    It follows that:
    \begin{proofsteps}{18em}
      \(\Gamma \models \at{k} \theta(\gamma)\)
      & since \(\Gamma \models \at{k} \forAll{X} \gamma\)
      \\
      \((W^{\Gamma_{0}}, M^{\Gamma_{0}}) \models \at{k} \theta(\gamma)\)
      & by the induction hypothesis
      \\
      \((W^{\Gamma_{0}}, M^{\Gamma_{0}}) \models^{w} \theta(\gamma)\), where \(w = W^{\Gamma_{0}}_{k}\)
      & by the definition of \(\models\)
      \\
      \label{proposition:init-quant-pc1}%
      \((W, M) \models^{w} \gamma\)
      & by the local satisfaction condition for \(\theta\)
      \\
      \((W^{\Gamma_{0}}, M^{\Gamma_{0}}) \models^{w} \forAll{X} \gamma\)
      & by \ref{proposition:init-quant-pc1}, since \((W, M)\) is an arbitrary \(\Delta[X]\)-expansion of \((W^{\Gamma_{0}}, M^{\Gamma_{0}})\)
      \\
      \((W^{\Gamma_{0}}, M^{\Gamma_{0}}) \models \at{k} \forAll{X} \gamma\)
      & by the definition of \(\models\)
    \end{proofsteps}

  \item[\(\nec{\act} \gamma\)]
    Suppose \(\Gamma \models \at{k} \nec{\act} \gamma\) and let \((w, w') \in W^{\Gamma_{0}}_{\act}\) such that \(w = W^{\Gamma_{0}}_{k}\).
    Since \((W^{\Gamma_{0}}, M^{\Gamma_{0}})\) is reachable, there exists a nominal term \(k'\) such that \(w' = W^{\Gamma_{0}}_{k'}\).
    It follows that:
    \begin{proofsteps}{18em}
      \label{proposition:init-nec-pc1}%
      \((W^{\Gamma_{0}}, M^{\Gamma_{0}}) \models \act(k, k')\)
      & since \((w, w') \in W^{\Gamma_{0}}_{\act}\)
      \\
      \(\Gamma_{0} \models \act(k, k')\)
      & by Corollary~\ref{corollary:HDCLS0-initiality}
      \\
      \(\Gamma \models \act(k, k')\)
      & since \(\Gamma \models \Gamma_{0}\)
      \\
      \(\Gamma \models \at{k'} \gamma\)
      & since \(\{\at{k} \nec{\act} \gamma, \act(k, k')\} \models \at{k'} \gamma\)
      \\
      \((W^{\Gamma_{0}}, M^{\Gamma_{0}}) \models \at{k'} \gamma\)
      & by the induction hypothesis
      \\
      \label{proposition:init-nec-pc2}%
      \((W^{\Gamma_{0}}, M^{\Gamma_{0}}) \models^{w'} \gamma\)
      & since \(w' = W^{\Gamma_{0}}_{k'}\)
      \\
      \((W^{\Gamma_{0}}, M^{\Gamma_{0}}) \models \at{k} \nec{\act} \gamma\)
      & from \ref{proposition:init-nec-pc1}--\ref{proposition:init-nec-pc2}, by the definition of \(\models\)
    \end{proofsteps}

  \item[\(\lnext{\sigma} \gamma\)]
    This case is trivial, because \(\at{k} \lnext{\sigma} \gamma \bimodels \at{\sigma(k)} \gamma\). 
    \qedhere
  \end{proofcases}
\end{proof}

Proposition~\ref{proposition:init} shows that the semantic consequences of a set \(\Gamma\) of Horn clauses are satisfied by the initial models of the set \(\Gamma_{0}\) of atomic sentences.
This is a very useful insight into the initiality property, which can be explored for other logical systems as well.

\begin{theorem} [Horn-clause initiality]
  \label{theorem:HDCLS-initiality}
  Let\/ \(\Gamma\) be a set of Horn clauses over a signature \(\Delta\) and \((W^{\Gamma_{0}}\!, M^{\Gamma_{0}})\) a reachable initial model of\/ \(\Gamma_{0}\).
  Then \((W^{\Gamma_{0}}\!, M^{\Gamma_{0}})\) is also an initial model of\/ \(\Gamma\).
\end{theorem}
\begin{proof}
  Since \(\Gamma \models \Gamma_{0}\), it follows that for every model \((W, M)\) of \(\Gamma\) there exists a unique homomorphism \((W^{\Gamma_{0}}, M^{\Gamma_{0}}) \to (W, M)\).
  Therefore, what remains to prove is that \((W^{\Gamma_{0}}, M^{\Gamma_{0}})\) is a model of \(\Gamma\).
  Let \(w \in |W^{\Gamma_{0}}|\) and \(\gamma \in \Gamma\).
  Since \((W^{\Gamma_{0}}, M^{\Gamma_{0}})\) is reachable, there exists a nominal term \(k\) such that \(w = W^{\Gamma_{0}}_{k}\).
  Moreover, since \(\Gamma \models \gamma\), we have \(\Gamma \models \at{k} \gamma\).
  By Lemma~\ref{proposition:init}, we deduce that \((W^{\Gamma_{0}}, M^{\Gamma_{0}}) \models \at{k} \gamma\), which is equivalent to \((W^{\Gamma_{0}}, M^{\Gamma_{0}}) \models^{w} \gamma\).
\end{proof}

%%%%%%%%%%%%%%%%%%%%%%%%%%%%%%%%%%%%%%%%%%%%%%%%%%%%%%%%%%%%%%%%%%%%%%%%%%%%%%%% 

\section{Herbrand's theorem}
\label{section:Herbrand-theorem}

A query for a signature \(\Delta\) is a sentence of the form \(\Exists{X} \bigand E\) such that \(X\) is a set of variables and \(E\) is a finite set of atomic sentences or action relations.
The following result, which is fundamental for the development of logic programming for \(\HDFOLS\), states that establishing the (semantic) entailment of a query by a program -- understood here as a set of Horn clauses -- is equivalent to finding a (syntactic) \emph{correct-answer substitution} for that query.

\begin{theorem} [Herbrand's theorem]
  \label{theorem:Herbrand}
  Consider a \(\HDFOLS\)-signature \(\Delta\), a set\/ \(\Gamma\) of Horn clauses, and a query \(\Exists{X} \bigand E\) over \(\Delta\).
  The following statements are equivalent:
  \begin{enumerate}
  \item \label{LP1} \(\Gamma \models \Exists{X} \bigand E\).

  \item \label{LP2} \((W^{\Gamma}, M^{\Gamma}) \models \Exists{X} \bigand E\), where \((W^{\Gamma}, M^{\Gamma})\) is an initial model of\/ \(\Gamma\).

  \item \label{LP3} \(\Gamma \models \theta(E)\) for some \(\Delta\)-substitution \(\theta \colon X \to \emptyset\).
  \end{enumerate}
\end{theorem}
\begin{proof}
  Note that, by Theorem~\ref{theorem:HDCLS-initiality}, \(\Gamma\) has indeed a \emph{reachable} initial reachable model \((W^{\Gamma}, M^{\Gamma})\).
  \begin{proofcases}
  \item[\ref{LP1}~\(\implies\)~\ref{LP2}] By the fact that \((W^{\Gamma}, M^{\Gamma}) \models \Gamma\).

  \item[\ref{LP2}~\(\implies\)~\ref{LP3}]
    Let \(w\) be an arbitrary possible world of \((W^{\Gamma}, M^{\Gamma})\).
    It follows that:
    \begin{proofsteps}{21em}
      \((W, M) \models^{w} E\)
      \newline for some \(\Delta[X]\)-expansion \((W, M)\) of \((W^{\Gamma}, M^{\Gamma})\)
      & by the definition of \(\models\)
      \\
      \((W, M) = (W^{\Gamma}, M^{\Gamma}) \red_{\theta}\)
      \newline for some \(\Delta\)-substitution \(\theta \colon X \to \emptyset\)
      & since \((W^{\Gamma}, M^{\Gamma})\) is reachable
      \\
      \((W^{\Gamma}, M^{\Gamma}) \models \theta(E)\)
      & by the local sat.\ cond.\ for \(\theta\)
      \\
      \((W', M') \models \theta(E)\)
      \newline
      for all models \((W', M')\) such that \((W', M') \models \Gamma\)
      & by Fact~\ref{fact:sat-atoms-actrel}, based on the universal property of \((W^{\Gamma}, M^{\Gamma})\)
      \\
      \(\Gamma \models \theta(E)\)
      & by the definition of \(\models\)
    \end{proofsteps}

  \item[\ref{LP3}~\(\implies\)~\ref{LP1}]
    Let \((W, M)\) be a model of \(\Gamma\) and \(w \in |W|\).
    It follows that:
    \begin{proofsteps}{21em}
      \((W, M) \models^{w} \theta(E)\)
      & since \(\Gamma \models \theta(E)\)
      \\
      \((W, M) \red_{\theta} \models^{w} E\)
      & by the local sat.\ cond.\ for \(\theta\)
      \\
      \((W, M) \models^{w} \Exists{X} \bigand E\)
      & since \((W, M) \red_{\theta} \red_{\Delta} = (W, M)\)% is a \(\Delta[X]\)-expansion of \((W, M)\)
      \\
      \(\Gamma \models \Exists{X} \bigand E\)
      & by the definition of \(\models\)
      \qedhere
    \end{proofsteps}
  \end{proofcases}
\end{proof}
\vspace{-\smallskipamount}

%%%%%%%%%%%%%%%%%%%%%%%%%%%%%%%%%%%%%%%%%%%%%%%%%%%%%%%%%%%%%%%%%%%%%%%%%%%%%%%% 

\section{Related work}
\label{section:related-work}

This work builds on previous results presented by the first author in~\cite{Gaina17Her}, and is related to developments in the initial semantics for hybridized institutions reported in~\cite{Diaconescu16} and in the foundations of logic-independent logic programming from~\cite{Diaconescu04,TutuF15,TutuF17}.
Clearly, those related studies are all logic-independent, whereas here we made a deliberate choice to focus on a concrete logical system -- still, many of the results developed in this paper can be lifted to an institution-independent setting.
The hybrid-dynamic logic that we propose has a number of distinguishing features that take it outside the scope of the frameworks presented in~\cite{Gaina17Her,Diaconescu16,Diaconescu04}.
More specifically, since possible worlds are no longer plain, but have an algebraic structure, one cannot employ the strategy used in~\cite{Gaina17Her} (or~\cite{Diaconescu16}, for that matter) for proving Theorem~\ref{theorem:HDCLS0-initiality}.
In fact, the nominal framework discussed therein can be regarded as an instance of the one we consider here for \(\HDFOLS\), where all nominals are constants.
That restriction leads to a much simpler construction of the quotient models (and proof of the existence of initial models), where the quotienting on possible worlds can be done by means of signature morphisms whose action on nominal constants is surjective.

Compared to~\cite{Diaconescu16}, the initiality results presented in this paper are based on constructive arguments that avoid the heavier model-theoretic infrastructure of quasi-varieties and inclusion systems.
Moreover, those arguments hold as well for some classes of constrained models (e.g.\ models with reachable possible worlds) that do not form a quasi-variety.

Besides these technical advancements, it is also worth noting that the Horn clauses and queries examined in this paper exceed the expressive power of the clauses and queries considered in \cite{Gaina17Her}.
For instance, the logical implications used in clauses are no longer restricted to atoms; instead, the hypothesis of an implication can involve both atomic sentences and action relations (which in general do not have an initial model), and the conclusion can be any Horn clause.
Moreover, queries can involve action relations too, which in general are not basic sentences.
For that reason, the variant of Herbrand's theorem presented here does not fit into the framework advanced in~\cite{Diaconescu04}.
A more general approach to Herbrand's theorem can be found in~\cite{TutuF17}, but the result presented there relies on a much more complex institution-theoretic infrastructure, which is outside the scope of this paper.

%%%%%%%%%%%%%%%%%%%%%%%%%%%%%%%%%%%%%%%%%%%%%%%%%%%%%%%%%%%%%%%%%%%%%%%%%%%%%%%% 

\section{Conclusions and further work}
\label{section:conclusions}

In this work, we have presented a new hybrid-dynamic logic for the specification and analysis of reconfigurable systems.
From a semantic perspective, this allows us to capture reconfigurable systems as Kripke structures whose possible worlds%
\begin{inparenum}
  \inparitem have an algebraic structure, which provides support for operations on system configurations, and
  \inparitem are labelled with constrained first-order models that capture the local/inner structure and behaviour of configurations.
\end{inparenum}
The model constraints that we consider here are given by rigid declarations of first-order symbols (sorts, or operation or relation symbols), which ensure a uniform interpretation of those symbols across the possible worlds.
This kind of modelling is also reflected in the syntax of the logic, where we use nominal and hybrid terms to refer to possible worlds and to the elements of the first-order structures associated to those worlds.
Terms are used to form nominal and hybrid equations, as well as relational atoms, from which we build complex sentences using Boolean connectives, quantifiers, and operators that are specific to hybrid logics and to dynamic logics.
To the best or our knowledge, there is currently no other logical system that supports such an approach to the reconfiguration paradigm.

We have shown that every set of atomic hybrid-dynamic sentences admits a reachable initial model (Kripke structure), and then generalized this result to Horn clauses, which in this case have a highly complex structure, involving implications and universal quantifiers (as in first-order logic), but also hybrid-logic operators such as \emph{retrieve} and \emph{store}, and dynamic logic operators such as \emph{necessity} over actions.
This is a fundamental result of great importance, because it allows us to deal with hybrid-dynamic logic programs (formalized as sets of Horn clauses).
To that end, we have studied a variant of Herbrand's theorem -- which is central to the logic-programming paradigm -- for hybrid-dynamic first-order logic.

This work is part of a much broader research agenda on the rigorous formal development of reconfigurable systems.
There are several clear tasks that we aim to pursue further.
One of those is to generalize the results obtained here to an institution-independent setting so that we could apply them with ease to other logical systems, such as those of order-sorted algebra and of preordered algebra.
The most convenient way to achieve that is by taking into account the relationship between the basic level and the upper level discussed in Section~\ref{section:initiality}; the former is necessarily logic-dependent, whereas the latter can be developed for arbitrary logics.
A second important task is to study Birkhoff completeness and resolution-based inference rules for hybrid-dynamic logics, for which we aim to follow~\cite{Gaina17Bir,TutuF17}.
This will provide a minimal framework for the analysis of reconfigurable systems, and a foundation for developing appropriate tool support and a number of concrete case studies.

%%%%%%%%%%%%%%%%%%%%%%%%%%%%%%%%%%%%%%%%%%%%%%%%%%%%%%%%%%%%%%%%%%%%%%%%%%%%%%%% 

\bibliography{HDFOL}

%%%%%%%%%%%%%%%%%%%%%%%%%%%%%%%%%%%%%%%%%%%%%%%%%%%%%%%%%%%%%%%%%%%%%%%%%%%%%%%% 

\end{document}

%% file: preamble.tex
\newcommand{\minisec}[1]{%
  \par\addvspace{\smallskipamount}\noindent%
  \textbf{\sffamily #1}\enspace%
}

\usepackage{multicol}
\newcounter{InParEnum}
\newenvironment{inparenum}{%
  \noindent\ignorespaces%
  \setcounter{InParEnum}{0}%
}{%
  \ignorespacesafterend%
}
\newcommand{\inparitemformat}[1]{%
  \textsf{(\arabic{#1})}%
}
\newcommand{\inparitem}{%
  \stepcounter{InParEnum}%
  \inparitemformat{InParEnum}~%
}

\usepackage{amsmath}
\usepackage{amsthm}
\usepackage{amssymb}
\usepackage{mathabx}
\usepackage{mathtools}
\usepackage{xy}
\xyoption{all}

\usepackage{scalerel}

\newcommand{\bbsemicolon}{%
  \scalerel*{%
    \hbox{\usefont{U}{bbold}{m}{n} ;}%
  }{;}%
}

\AtBeginDocument{%
  \let\Dashv\relax
  \newcommand{\Dashv}{%
    \mathrel{\text{\reflectbox{\(\vDash\)}}}%
  }%
}

\newcommand{\vDashv}{%
  \mathrel{\text{\ooalign{\(\vDash\)\cr\(\Dashv\)\cr}}}%
}

\newtheorem{fact}[definition]{Fact}

\newcommand{\Sen}{\mathtt{Sen}}

\renewcommand{\models}{\vDash}
\newcommand{\nmodels}{\nvDash}
\newcommand{\bimodels}{\vDashv}

\newcommand{\comp}{\mathbin{\bbsemicolon}}
\newcommand{\red}{\mathord{\upharpoonright}}

\newcommand{\alt}{\enspace|\enspace}

\newcommand{\HDFOLS}{{\mathsf{HDFOLS}}}
\newcommand{\HDCLS}{{\mathsf{HDCLS}}}

\newcommand{\keyscript}[1]{\text{\normalfont\sffamily\itshape #1}}

\newcommand{\nominal}{\keyscript{n}}
\newcommand{\rigid}{\keyscript{r}}
\newcommand{\flexible}{\keyscript{f}}

\newcommand{\assign}{\leftarrow}

\newcommand{\ari}{\mathit{a\mkern -1mu r}}
\newcommand{\act}{\mathfrak{a}}

\newcommand{\bigand}{\bigwedge}
\renewcommand{\implies}{\Rightarrow}
\newcommand{\at}[1]{@_{#1}\,}
\newcommand{\nec}[1]{[#1]}
\newcommand{\store}[1]{{\downarrow}#1\,{\cdot}\,}
\newcommand{\forAll}[1]{\forall #1\,{\cdot}\,}
\newcommand{\Exists}[1]{\exists #1\,{\cdot}\,}
\newcommand{\lnext}[1]{(#1)\,}

\usepackage{mathpartir}

\usepackage{calc}
\usepackage{longtable}
\usepackage{booktabs}
\newcommand{\tabStretch}[1]{\renewcommand{\arraystretch}{#1}}
\newcommand{\tabColSep}[1]{\setlength{\tabcolsep}{#1}}
\usepackage{array,etoolbox}
\preto\tabular{\setcounter{tablerow}{0}}
\newcounter{tablerow}
\newcommand\rownumber{%
  \stepcounter{tablerow}\thetablerow\addtocounter{tablerow}{-1}%
}

\newenvironment{proofsteps}[1]{%
  \begingroup%
  \setlength\LTpre{\smallskipamount}%
  \setlength\LTpost{\smallskipamount}%
  \setcounter{tablerow}{0}%
  \tabStretch{1.2}\tabColSep{0em}\noindent%
  \longtable{%
    @{\makebox[\widthof{\footnotesize 9}][r]{\footnotesize\rownumber}}%
    @{\hskip 1.5em}>{\refstepcounter{tablerow}}p{#1}%
    @{\hskip 1.5em}>{\footnotesize\raggedright\arraybackslash}p{\linewidth-#1-4em}%
  }%
}{%
  \endlongtable%
  \addtocounter{table}{-1}%
  \endgroup%
}

\newenvironment{proofcases}{%
  \renewcommand{\descriptionlabel}[1]{%
    \hspace\labelsep\normalfont\sffamily [\;##1\;]%
  }%
  \description%
}{\enddescription}